\documentclass[sigconf]{acmart}
\usepackage{braket}
\usepackage[ruled,vlined,linesnumbered,noresetcount]{algorithm2e}
\usepackage{multirow}
\usepackage{subcaption} 
\usepackage{algorithmic}
\usepackage{threeparttable}
\usepackage{amsthm}
\usepackage{mathtools}

\copyrightyear{2024}
\acmYear{2024}
\setcopyright{acmlicensed}\acmConference[ICCAD '24]{IEEE/ACM International Conference on Computer-Aided Design}{October 27--31, 2024}{New York, NY, USA}
\acmBooktitle{IEEE/ACM International Conference on Computer-Aided Design (ICCAD '24), October 27--31, 2024, New York, NY, USA}
\acmDOI{10.1145/3676536.3676678}
\acmISBN{979-8-4007-1077-3/24/10}

\AtBeginDocument{%
  }

\newcommand{\nprogram}{K}
\newcommand{\nqpu}{N}
\newcommand{\nqubit}{n}

\begin{document}

\title{On Reducing the Execution Latency of Superconducting Quantum Processors via Quantum Job Scheduling}

\author{Wenjie Wu}
\affiliation{%
  \institution{Shanghai Jiao Tong University}
  \city{Shanghai}
  \country{China}
}
\email{wenjiewu@sjtu.edu.cn}

\author{Yiquan Wang}
\affiliation{
  \institution{Shanghai Jiao Tong University}
  \city{Shanghai}
  \country{China}
}
\email{abcdfehg@sjtu.edu.cn}

\author{Ge Yan}
\affiliation{
  \institution{Shanghai Jiao Tong University}
  \city{Shanghai}
  \country{China}
}
\email{yange98@sjtu.edu.cn}

\author{Yuming Zhao}
\affiliation{
  \institution{Shanghai Jiao Tong University}
  \city{Shanghai}
  \country{China}
}
\email{arola_zym@sjtu.edu.cn}

\author{Bo Zhang}
\affiliation{
  \institution{Shanghai AI Laboratory}
  \city{Shanghai}
  \country{China}
}
\email{bo.zhangzx@gmail.com}

\author{Junchi Yan}
\authornote{Corresponding author. The work was partly supported by NSFC (92370201) and QuantumCtek Quantum Cloud Services.}
\affiliation{
  \institution{Shanghai Jiao Tong University}
  \city{Shanghai}
  \country{China}
}
\email{yanjunchi@sjtu.edu.cn}

\renewcommand{\shortauthors}{Wu W, Wang Y, Yan G, et al.}

\begin{abstract}
Quantum computing has gained considerable attention, especially after the arrival of the Noisy Intermediate-Scale Quantum (NISQ) era. Quantum processors and cloud services have been made world-wide increasingly available. Unfortunately, jobs on existing quantum processors are often executed in series, and the workload could be heavy to the processor. Typically, one has to wait for hours or even longer to obtain the result of a single quantum job on public quantum cloud due to long queue time. In fact, as the scale grows, the qubit utilization rate of the serial execution mode will further diminish, causing the waste of quantum resources. In this paper, to our best knowledge for the first time, the Quantum Job Scheduling Problem (QJSP) is formulated and introduced, and we accordingly aim to improve the utility efficiency of quantum resources. Specifically, a noise-aware quantum job scheduler (NAQJS) concerning the circuit width, number of measurement shots, and submission time of quantum jobs is proposed to reduce the execution latency. We conduct extensive experiments on a simulated Qiskit noise model, as well as on the Xiaohong (from QuantumCTek) superconducting quantum processor. Numerical results show the effectiveness in both the QPU time and turnaround time.  
\end{abstract}

\ccsdesc[500]{Hardware~Quantum computation}
\keywords{Quantum Computing, Quantum Job Scheduling, Quantum Cloud}


\maketitle
\section{Introduction}
In recent decades, considerable progress has been made in quantum computing (QC). Shor's algorithm \cite{shor1994algorithms} achieves exponential acceleration for factor decomposition, and Grover's algorithm \cite{grover1996fast} provides quadratic speedup for unstructured search over classical counterparts. Recently, the development of quantum computers and methods has led us into the so-called Noisy Intermediate-Scale Quantum (NISQ) era \cite{preskill2018quantum}, with some evidence on the so-called quantum supremacy, e.g. Google's superconducting quantum processor Sycamore~\cite{arute2019quantum}. The potential advantage of QC over classical computing are attracting increasing attention. 


More and more players like IBM have provided the public access to their quantum computers. This facilitates the validation of quantum algorithms on NISQ devices over the Internet. For example, we have free access to the 7-qubit IBM Perth \cite{ibmquantum}. However, running quantum circuits on current quantum computers is non-trivial due to the noise and sparse connectivity of physical qubits. On a NISQ device, the physical qubits are not fully connected. The deployment of two-qubit gates is restricted to pairs of connected qubits. Hence, when mapping logical qubits to their physical counterparts, certain two-qubit gates may be positioned on physically disconnected qubits, rendering them inexecutable. Conventionally, SWAP gates are inserted to change the qubit mapping so that every two-qubit gate can be physically executed. Since SWAP gates result in extra noise, the number of them is expected to be minimized. 

A more awkward obstacle hindering people from using quantum computers is the unbearably long queue time. Though there exist some quantum cloud services, the growing need for quantum hardware outpaces the open access to quantum hardware. To verify this, we submit 20 jobs to IBM Perth within a week. According to the panel, the average number of pending jobs when submitting is about 2,540, and the average queue time before execution is about 6.7 hours. The latency of circuit execution is unacceptable, especially when we run Variational Quantum Algorithms (VQAs) \cite{cerezo2021variational}, in which plenty of circuits are executed in a single episode to update the parameters. The main reason for this latency is that the submitted quantum jobs are executed in series. Thus, only one job is executed on the quantum processor in each execution. Besides, entangling a large number of qubits on NISQ devices is challenging due to the noise \cite{cao2023generation}, so most circuits remain small in width to ensure high fidelity. Hence, the qubit utilization rate is low. With the increasing number of physical qubits on QPUs and decreasing error rate, we may execute multiple jobs in parallel in each execution, i.e. quantum multi-programming (QMP), at a negligible cost of fidelity to reduce the latency. Such parallel running mode can be extended to applications like quantum architecture search \cite{wu2023quantumdarts,lu2023qas}, quantum inner product \cite{Xiong_2024_CVPR}, and network alignment \cite{ye2023vqne}. As a result, more people can access quantum resources to facilitate QC. 
\begin{figure}[t]
\includegraphics[width=\linewidth]{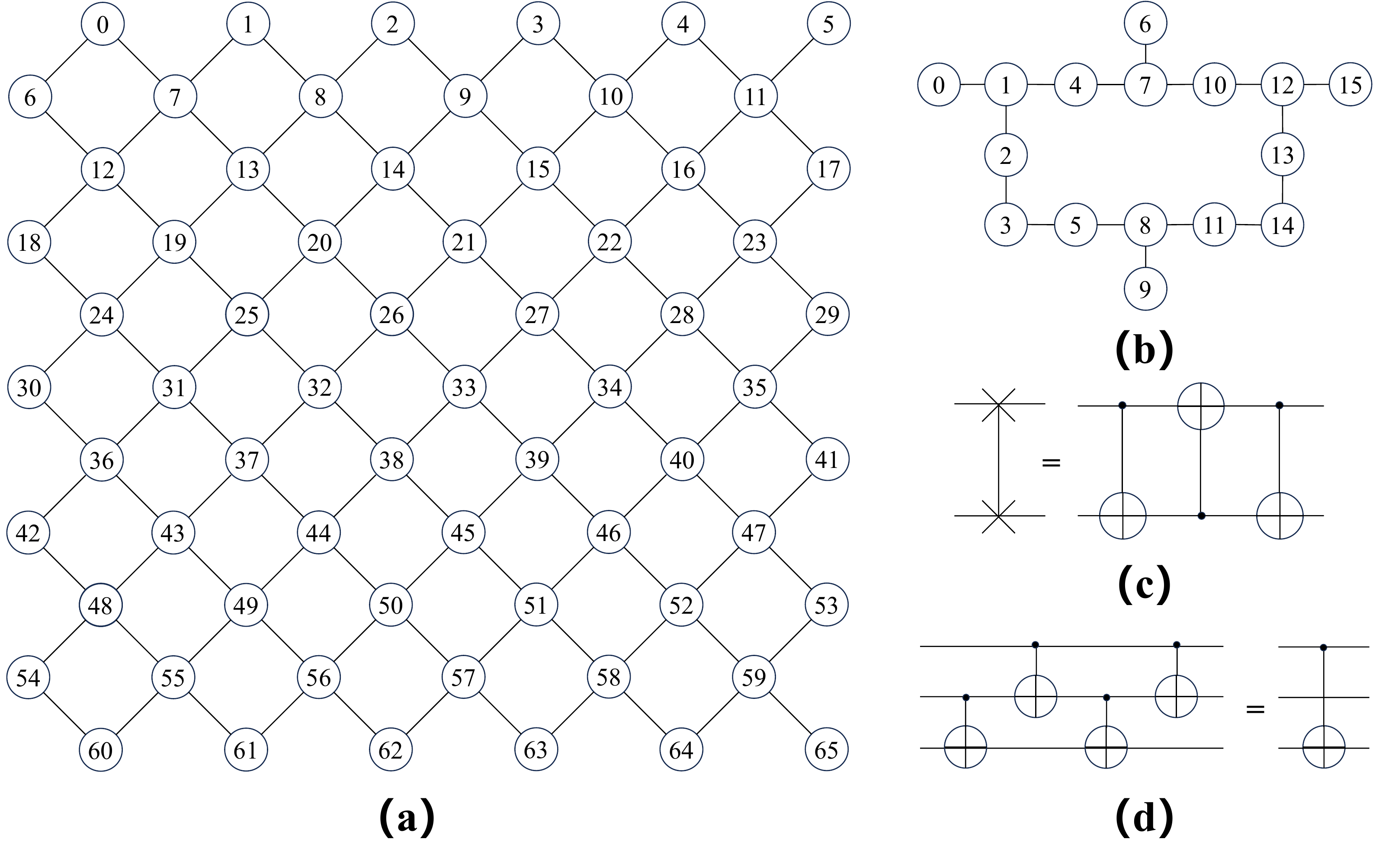}
\caption{(a) Coupling graph of Xiaohong quantum processor (from QuantumCTek as used in this paper for experiments). (b) Coupling graph of IBM Guadalupe. (c) SWAP gate. (d) BRIDGE gate. SWAP and BRIDGE gates can solve the connectivity constraints on coupling graphs.}\label{fig:qpu}\vspace{-5pt}
\end{figure}

\begin{figure*}
\centering
	\begin{subfigure}[b]{0.09\linewidth}
	\includegraphics[width=\linewidth]{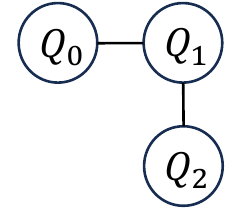}
	\caption{}
\label{fig:subgraph}
	\end{subfigure}
	\begin{subfigure}[b]{0.23\linewidth}
    \includegraphics[width=\linewidth]{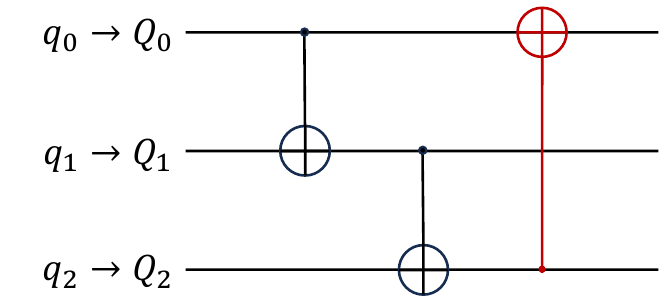}
	\caption{}
\label{fig:violate}
\end{subfigure}  
	\begin{subfigure}[b]{0.33\linewidth}
	\includegraphics[width=\linewidth]{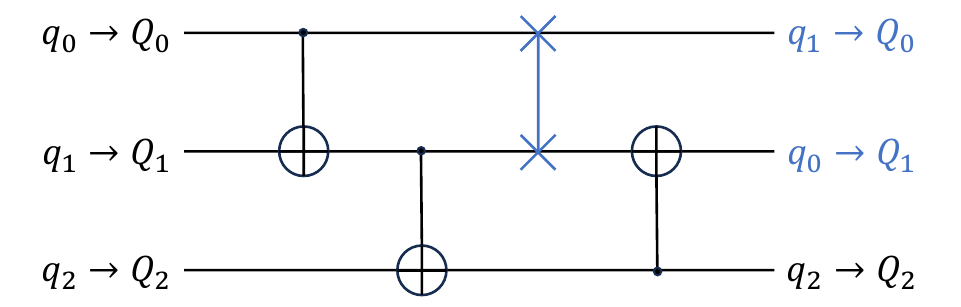}
	\caption{}
\label{fig:qmap_swap}
	\end{subfigure}
	\begin{subfigure}[b]{0.33\linewidth}
    \includegraphics[width=\linewidth]{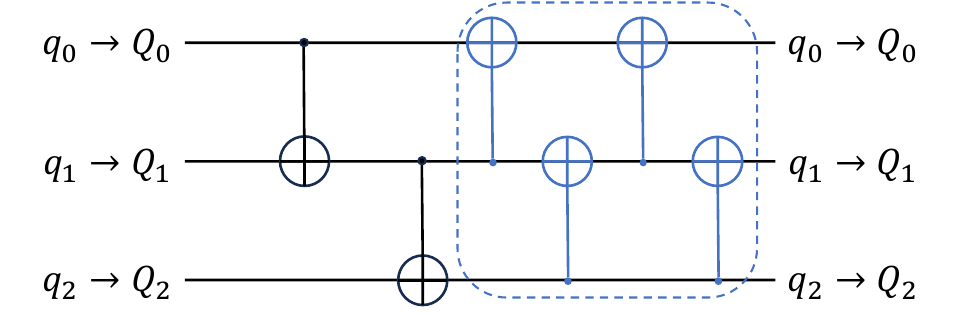}
	\caption{}
\label{fig:qmap_bridge}
\end{subfigure}
\vspace{-10pt}
\caption{An example of qubit mapping. (a) Subgraph derived by qubit partitioning. (b) Quantum circuit to be mapped. (The CNOT gate in red cannot be applied, because $Q_0$ and $Q_2$ are not connected.) (c) Mapped circuit through SWAP gates. (d) Mapped circuit through BRIDGE gates. The SWAP gate changes the mapping in (c) (marked in blue).}\label{fig:qmap}
\vspace{-10pt}
\end{figure*}

However, QMP on quantum processors is a complicated task. The execution order of circuits will affect its performance. Different from classical process scheduling, we need to consider fidelity apart from time metrics. The QPU should be partitioned in a fair manner to reduce fidelity drop. Unfortunately, fidelity and time metrics often conflict with each other. In this paper, we introduce the Quantum Job Scheduling Problem, which has great practical value in the NISQ era. A novel scheduling method is proposed to tackle this problem. With our priority score and noise-aware initial mapping, our method surpasses baselines in time metrics, and guarantees the fairness and fidelity. \textbf{The contributions of this paper are:}

1) We formulate the Quantum Job Scheduling Problem of reducing the latency of (superconducting) quantum processors, fully utilizing the computational power of quantum processors.

2) We propose a novel noise-aware quantum job scheduler to balance time metrics, fidelity, and fairness. The small overhead caused by our method can be neglected.

3) Experimental results on both the noise model and real-world quantum computer show that our approach significantly reduces the QPU time and turnaround time at a low cost of fidelity. 
\begin{figure}
    \vspace{-5pt}
    \centering
    \includegraphics[width=\linewidth]{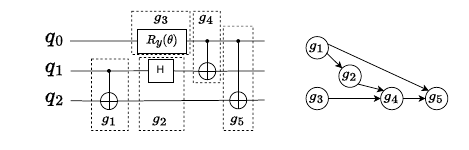}
    \vspace{-25pt}
    \caption{A quantum circuit (left) and its directed acyclic graph (right).}\vspace{-15pt}
    \label{fig:dag}
\end{figure}

\section{Preliminaries and Related Works}
We discuss some basic concepts and loosely related works to ours. To our best knowledge, there still does not exist peer methods for the scheduling problem addressed in our paper.

\textbf{Quantum Computing.}
The basic unit in QC is a qubit, which is in superposition of basis states $\ket{0}$ and $\ket{1}$: $\ket{\psi}=a\ket{0} + b\ket{1}$, where $|a|^2+|b|^2=1$. Likewise, a quantum system with $n$ qubits is in superposition of $2^n$ basis states. The evolution of quantum states can generate solutions to specific problems, perhaps much faster than classical methods. We refer readers to \cite{nielsen2010quantum} for detailed backgrounds. Quantum circuits are employed to implement quantum computation of quantum states. Each quantum circuit consists of quantum gates like X gates, RZ gates, CNOT gates, etc. To obtain the result, we have to repeat executing the circuit many times (shots), because the quantum measurement will cause the collapse of a superposition state to a basis state. A three-qubit quantum circuit is given in Fig. \ref{fig:dag} as an example. A quantum circuit can further be converted into a directed acyclic graph (DAG). The topological order of the DAG corresponds with the execution order of quantum gates (from left to right). For example, gate $g_4$ cannot be executed until gate $g_1$, $g_2$ and $g_3$ are executed in Fig. \ref{fig:dag}. 

\textbf{Quantum Processors.} The core of a quantum computer is the quantum processor, aka QPU, which serves to execute quantum circuits. We focus on superconducting quantum processors in this paper. The major properties of a QPU are its basis gates, coupling graph and noise condition. Basis gates are the quantum gates supported on the QPU. All the gates in a quantum circuit must be converted to combinations of basis gates during compiling before execution. As shown in Fig. \ref{fig:qpu}, the coupling graph restricts the connectivity of qubits. Two-qubit gates can only be deployed on connected qubits. Besides, the noise of QPUs in the NISQ era results in gate errors, measurement (readout) errors, and decoherence, which will corrupt the quantum state and reduce the fidelity. These errors change over time, so they must be calibrated regularly. Nowadays, many quantum processors are open to public through quantum cloud services. Our submitted quantum jobs will queue up to be executed. In this paper, we conduct experiments on the Qiskit \cite{aleksandrowicz2019qiskit} noise model of 16-qubit IBM Guadalupe (Fig. \ref{fig:qpu}b), and 66-qubit Xiaohong\footnote{As used in our experiments, Xiaohong is a 66-qubit superconducting quantum processor, which can be accessed via public cloud at \url{https://quantumctek-cloud.com/}. The used QCIS instruction set can be easily converted from or to the widely used QASM.} (Fig. \ref{fig:qpu}a) quantum processor from QuantumCTek \cite{quantumctek}.

\textbf{Qubit Mapping.}
When logical qubits of a quantum circuit are mapped to physical qubits on a QPU, the original two-qubit gates may violate the connectivity constraints as shown in Fig. \ref{fig:violate}. A traditional way to solve this problem is to insert SWAP gates. A SWAP gate is implemented by three CNOT gates (Fig. \ref{fig:qpu}c), incurring extra noise. Hence, the number of them is expected to be minimized. Siraichi \textit{et al.} formally introduce the aforementioned qubit allocation (mapping) problem \cite{siraichi2018qubit}, which is proved to be NP-complete. Li \textit{et al.} propose a bidirectional heuristic search (SABRE) to tackle this problem. When inserting a SWAP gate, they consider its impact on two-qubit gates in both the front layer and extended set, significantly reducing the SWAP overhead \cite{li2019tackling}. Niu \textit{et al.} take the error rate and execution time of CNOT gates into consideration, and provide BRIDGE gates as an alternative to SWAP gates \cite{niu2020hardware}. The BRIDGE gate (Fig. \ref{fig:qpu}d) is composed of four CNOT gates, but its effect equals a single CNOT gate, without changing the mapping. Niu \textit{et al.} further ameliorate the mapping method by involving the cost of inserted SWAP gates and BRIDGE gates themselves \cite{niu2023enabling}. Other methods like Reinforcement Learning (RL) \cite{huang2022reinforcement}, Monte Carlo Tree Search (MCTS) \cite{zhou2020monte,sinha2022qubit}, binary integer programming \cite{nannicini2022optimal}, and Satisfiability
Modulo Theory (SMT) \cite{murali2019noise} have also been studied in qubit mapping. We refer to \cite{ge2024quantum} for a more comprehensive survey.

\textbf{Quantum Multi-Programming (QMP).}
QMP means running multiple quantum circuits simultaneously on a QPU. This task can be decomposed into two sub-tasks, i.e. qubit partition and qubit mapping. Qubit partition allocates a unique region on the QPU to every parallel quantum circuit. Then, qubit mapping pairs logical qubits with physical qubits on the partition, and inserts SWAP gates to satisfy all two-qubit constraints. Das \textit{et al.} propose QMP on NISQ devices to improve throughput \cite{das2019case}. They allocate less noisy physical qubits to logical qubits with higher utility. Qucloud \cite{liu2021qucloud} leverages FN community detection algorithm \cite{newman2004fast} to partition QPUs, and designs an EPST score to estimate the fidelity of allocation. Different quantum circuits can be executed together only when the gap between co-located EPST and separate EPST is less than the threshold. Resch \textit{et al.} run multiple QAOA \cite{farhi2014quantum} circuits in parallel to accelerate the training process \cite{resch2021accelerating}. They greedily expand the partition by breadth-first search (BFS) based on heuristics. All the three methods \cite{das2019case,liu2021qucloud,resch2021accelerating} utilize SABRE \cite{li2019tackling} to conduct qubit mapping. Niu \textit{et al.} reorder the quantum circuits according to their CNOT density, and partition QPUs based on the connectivity and error rates of physical qubits \cite{niu2023enabling}. These existing methods either disregard the execution order or just focus on QMP in single execution. 

\section{Methodology}
In this section, we formally introduce the Quantum Job Scheduling Problem (QJSP) to excavate the importance of the execution order when multi-programming massive quantum circuits in the queue of quantum cloud services. Also, a noise-aware quantum job scheduler (NAQJS) is proposed to tackle this problem.

\subsection{The Quantum Job Scheduling Problem}
\subsubsection{Definition}
Suppose the current job queue $\mathcal{Q}$ is comprised of $\nprogram$ quantum jobs to be executed, i.e. $\mathcal{Q}=\{\mathcal{J}_1, \mathcal{J}_1, \cdots , \mathcal{J}_{\nprogram}\}$. Each job $\mathcal{J}_i$ can be represented as a tuple $(c_i, s_i, t_i)$, where $c_i$, $s_i$, and $t_i$ denote the quantum circuit, number of measurement shots, and submission time, respectively. Then, new jobs $\mathcal{J}_{\nprogram+1}, \mathcal{J}_{\nprogram+2}, \cdots$ will be submitted at time $t_{\nprogram+1}, t_{\nprogram+2}, \cdots$. For a quantum computer, besides the execution time $t_e$ of circuits on the QPU, other procedures like circuit verification, generation of control signals, and communication will cost extra time $t_m$ between execution. 

Given the coupling graph $\mathcal{G}$, noise calibration data $\mathcal{N}$, and basis gate set $\mathcal{B}$ of the QPU, we need to execute all the jobs submitted during a time period on the QPU. The objective of QJSP is to minimize the execution latency of jobs, and maintain high fidelity.

\subsubsection{Metrics}
The performance assessment of QJSP is divided into two parts: time and fidelity. Also, fairness should be considered. 

\textbf{Time.} For users, they mainly care about the time cost from submission to completion of their quantum job, which we name turnaround time. For suppliers, they emphasize on the QPU time of their quantum processors, i.e. total circuit execution time on QPU. 

\textbf{Fidelity.} The real fidelity of a quantum state is hard to obtain on a quantum computer, because recovering the complete quantum state from measurements is non-trivial. Methods like classical shadow \cite{aaronson2018shadow,huang2020predicting} can mitigate this problem but still incur additional overhead on quantum processors to achieve ideal accuracy. By convention, we use the Probability of Successful Trial (PST), which is defined as the percentage of trials producing the correct result, as our fidelity metric. This metric is widely used in NISQ applications \cite{linke2017experimental,tannu2019not,das2019case,tannu2019mitigating,liu2021qucloud,niu2023enabling} as an alternative to fidelity for its cost-efficiency.
\begin{figure}[tb!]
  \centering
  \includegraphics[width=\linewidth]{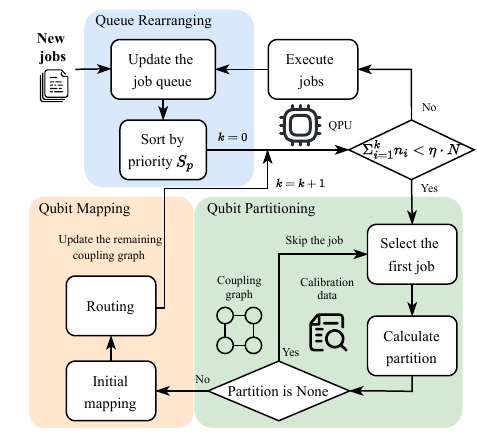}
  \caption{Overview of our noise-aware quantum job scheduler (NAQJS). }\label{fig:pipeline}
\end{figure}

\subsection{Proposed Method}
Our proposed method is composed of three parts: queue rearranging, qubit partitioning, and qubit mapping. Fig. \ref{fig:pipeline} shows the overview of NAQJS. In each iteration, we sort the current updated queue by our priority score $S_p$. Then, we evaluate quantum jobs in the sorted queue one by one, selecting and mapping those jobs whose circuits can find a partition on the remaining coupling graph, until the number of used physical qubits exceed the limit, i.e. $\Sigma_{i=1}^k n_i \leq \eta \cdot \nqpu$, where $n_i$ is the number of qubits (i.e. width) of the $i$-th selected jobs, and $k$ is the number of selected jobs. $\nqpu$ denotes the number of physical qubits, and $\eta\in (0,1]$ is the allowed maximum usage of physical qubits, which influences the average fidelity because higher usage will incur more noise. Also, $\eta$ can prevent jobs from using extremely noisy qubits. Then, the mapped jobs will be executed in parallel on the QPU. The number of shots is set as the maximum shot number among the mapped jobs to ensure all the shot requirements are satisfied, because those jobs whose shot requirements are unsatisfied will lead to extra execution overhead in following iterations. Since we can retain only the first $s_i$ outcomes, this execution pattern will not affect the result. The mapped circuits are executed in an As Late As Possible (ALAP) manner so that circuits with different depth can be measured and completed at the same time to avoid decoherence.  

\subsubsection{Queue Rearranging}
In each iteration, the job queue will be updated due to new submitted jobs and executed jobs. Akin to the importance of process scheduling for CPUs, the execution order of quantum jobs also counts in QJSP. Therefore, we sort quantum jobs in the updated queue in descending order of priority score. Three properties of a quantum job $\mathcal{J}_i$ are considered for its priority score $S_p^{(i)}$: the number of qubits $n_i$, number of shots $s_i$, and submission time $t_i$. The priority score is defined as the linear combination:
\begin{equation}\label{eq:priority}
    S_p^{(i)} = -\alpha \cdot S_n^{(i)} -\beta \cdot S_s^{(i)} - \gamma \cdot S_t^{(i)},
\end{equation}
where $\alpha$, $\beta$, and $\gamma$ ($\alpha,\beta,\gamma\geq 0$) denote the width weight, shot weight, and time weight. $S_n^{(i)}$, $S_s^{(i)}$, and $S_t^{(i)}$ are the Min-Max Normalization results of $n_i$, $s_i$, and $t_i$. For example, $S_n^{(i)}$ can be calculated as $S_n^{(i)} = (n_i - n_{min})/(n_{max}-n_{min})$.

The number of qubits $n_i$, number of shots $s_i$, and submission time $t_i$ are deemed as three most important factors for the time and fidelity metric in QJSP. The reason is as follows:

\textbf{\# Qubits.} When $n_i$ is small, the QPU can accommodate more jobs, which means more jobs are executed in unit time. This is similar to the Shortest Job First (SJF) strategy in process scheduling, which significantly raises the throughput at the beginning, thus improving the average turnaround time. 

\textbf{\# Shots.} The term $S_s^{(i)}$ narrows the distance of $s_i$ between neighboring jobs in the queue. Since the shot number is set as the maximum of $s_i$ among jobs in one execution, $S_s^{(i)}$ can decrease the number of unnecessary shots, which in turn reduces the QPU time. Also, a small $s_i$ accelerates the execution of quantum jobs at the beginning, which can benefit the turnaround time.

\textbf{Submission time.} The term $S_t^{(i)}$ prioritizes early-submit\-ted jobs, sacrificing the turnaround time and QPU time for fairness, which embodies in the maximum and standard deviation of turnaround time. 

It is worth mentioning that the circuit depth is excluded from the calculation of the priority score. For Xiaohong, the execution time of every shot is set as a constant (i.e. 0.2 ms) in reality, no matter how deep the executed circuit is. The constant execution time makes it convenient for the system to operate. Also, this time is long enough for both the execution of the deepest circuit allowed and qubit de-excitation. We have executed circuits of different depth on Xiaohong, and found that the execution time of every shot is about 0.2 ms for all circuits. Similar property is found on IBM quantum cloud empirically. Hence, the circuit depth makes no difference on latency in practical settings.

Alike process scheduling, QJSP also faces the starvation problem \cite{tanenbaum2009modern} that a quantum job waits infinitely long to run because its priority score is lower than others all the time. Starvation occurs when the number of qubits and shots of a quantum job is extremely large. The term $S_t^{(i)}$ mitigates the starvation problem but cannot avoid it. Hence, we adopt an aging strategy, i.e. raising the priority score $S_p^{(i)}$ by 1 every $\Delta t$ seconds when job $\mathcal{J}_i$ waits in the queue. 

\begin{theorem}
    With our aging strategy, all the quantum jobs can be executed in finite time in QJSP.
\end{theorem}
\begin{proof}
    For any job $\mathcal{J}_i$ after the Min-Max Normalization, $S_n^{(i)}$, $S_s^{(i)}$ and $S_t^{(i)}$ range from 0 to 1. Then, the priority score $S_p^{(i)} \in \left[-\alpha-\beta-\gamma, 0\right]$. When the time reaches $t^\prime \coloneqq t_i + (\alpha+\beta+\gamma+1)\Delta t$, its priority score satisfies $S_p^{(i)}\in \left[1, \alpha+\beta+\gamma+1\right]$, larger than that of any job submitted after $t^\prime$. Hence, any job submitted after $t^\prime$ will be executed later than $\mathcal{J}_i$. Since the number of jobs submitted before $t^\prime$ is finite, $\mathcal{J}_i$ will be executed in finite time.
\end{proof}
\begin{algorithm}[t]
\small
\SetAlgoLined     
\KwIn{DAG of the Circuit $G$, Number of Qubits $n$}
\KwOut{Circuit Time $t_c$}
  Initialize $t[i] = 0$, for $i = 0,1\cdots n-1$\;
  Initialize an empty queue $Q$\;
  \For{gate $g$ in $G$}{
    \If{$g.in\_degree == 0$}{
    $Q.push(g)$\;
    }
  }
  \While{$Q$ is not empty}{
    $g = Q.top()$\;
    \If{$g$ is one-qubit gate}{
    $q_1 = g.qubits$\;
    $t[q_1] = t[q_1] + g.duration$\;
    }
    \If{$g$ is two-qubit gate}{
    $q_1, q_2 = g.qubits$\;
    $t[q_1]=t[q_2]=max(t[q_1],t[q_2])+g.duration$\;
    }
    \For{gate $g^\prime$ in $g.successors$}{
        $g^\prime.in\_degree = g^\prime.in\_degree - 1$\;
        \If{$g^\prime.in\_degree == 0$}{
        $Q.push(g^\prime)$\;
        }
    }
    $Q.pop()$\;
    }
    $t_c = max(t[i])$\;

\caption{Calculation of circuit time}
\label{alg:circuit_time}
\end{algorithm}

\subsubsection{Qubit Partitioning}
After rearranging the updated queue, we need to select a number of jobs to be executed in parallel in this iteration. Selected jobs must share no common physical qubits with each other, so we should partition the coupling graph into separate parts. Concretely, we pick out the first job in the sorted queue to conduct qubit partitioning. If the partitioning algorithm cannot find a valid partition, it will skip the job, and proceed to the next. Otherwise, we will go on to the qubit mapping step. We use the qubit partitioning algorithm introduced in \cite{niu2023enabling}, because it considers both the noise and topology of the QPU, and substantially reduces the search space by limiting the starting points. 

This method chooses physical qubits with higher degrees than the largest logical degree as starting points. If such qubits do not exist, it chooses physical qubits with the highest degree as starting points. Then, it adds a neighboring qubit with the highest fidelity degree to the partition iteratively until the number of selected physical qubits equals that of logical qubits. The fidelity degree $D_f^{(i)}$ is calculated as $D_f^{(i)} = 2\times\sum\nolimits_{j\in N(i)} r_{2q}^{(i,j)} + r_{ro}^{(i)}$, where $N(i)$ denotes the neighboring qubits of $Q_i$. $r_{2q}^{(i,j)}$ and $r_{ro}^{(i)}$ are the reliability of two-qubit gates on edge $(Q_i, Q_j)$ and measurements (readout) on $Q_i$. Finally, the partition with the best fidelity score is selected. The fidelity score $S_f$ is derived from Eq.\eqref{eq:sf}: 
\begin{equation}\label{eq:sf}
    S_f = -N_{2q}\times (1-\overline{r}_{2q}) - N_{ro}\times (1-\overline{r}_{ro}),
\end{equation}
where $\overline{r}_{2q}$ and $\overline{r}_{ro}$ are average reliability of two-qubit gates and measurements in this partition. $N_{2q}$ and $N_{ro}$ is the number of the two operations. 

\subsubsection{Qubit Mapping}
This step is further divided into two sub-tasks: initial mapping and routing. Initial mapping is to determine the initial one-to-one correspondence between logical and physical qubits. \cite{siraichi2018qubit} shows that initial mapping can affect the final circuit quality. However, initial mapping alone cannot ensure the applicability of all the two-qubit gates. Then, routing solves the constraints of these two-qubit gates one by one. Finally, used physical qubits in the mapping are removed from the remaining coupling graph. An example of qubit mapping is given in Fig. \ref{fig:qmap}. The CNOT gate in red cannot be applied on the subgraph derived by qubit partitioning. Two solutions are provided. One is to insert a SWAP gate to exchange the state of logical qubits $q_0$ and $q_1$, so the CNOT gate should act on physical qubits $Q_1$ and $Q_2$ (Fig. \ref{fig:qmap_swap}). Another is to use BRIDGE gate to connect $Q_0$ and $Q_2$ via an intermediary qubit $Q_1$ (Fig. \ref{fig:qmap_bridge}). The difference is that SWAP gates will change the logical-to-physical mapping while BRDIGE gates keep it unchanged.

\textbf{Initial mapping}. \cite{li2019tackling} proposes a reverse traversal technique to refine initial mapping. A quantum circuit can be easily reversed, retaining the same connectivity constraints as the original one. Therefore, we can exploit the final mapping of the reverse circuit as the new initial mapping of the original circuit to improve the mapping result. Nevertheless, this method overlooks the impact of varying noise among qubits. \cite{liu2021qucloud} designs the $EPST$ score to estimate the probability of a successful trial under noise, but the score is calculated from the average reliability of gates and measurements, which may deviate from reality. Then we define the $EPST^*$ score:
\begin{equation}\label{eq:epst*}
    EPST^* = \prod_{i=1}^{N_{1q}}r_{1q}^{(o_i)}\cdot\prod_{i=1}^{N_{2q}}r_{2q}^{(d_i)}\cdot\prod_{i=1}^{N_{ro}}r_{ro}^{(m_i)}\cdot\prod_{i=1}^{n}r_{a}^{(i)}\cdot\prod_{i=1}^{n}r_{p}^{(i)},
\end{equation}
where $r_{1q}^{(*)}$, $r_{2q}^{(*)}$, and $r_{ro}^{(*)}$ denote the reliability of one-qubit gates, two-qubit gates, and measurements. $N_{1q}$, $N_{2q}$, and $N_{ro}$ is the number of the three operations. $o_i$, $d_i$, and $m_i$ map operations to their locations in the circuit. Since our $EPST^*$ score considers each gate's reliability separately, it is more accurate than $EPST$, especially given high variance of reliability. Moreover, $r_{a}^{(i)}$ and $r_{p}^{(i)}$, the probability of amplitude damping error and phase damping error not occurring on the i-th qubit, are involved in $EPST^*$ to perceive the impact of decoherence. They can be calculated as: $r_{a}^{(i)} = \exp{(-t_c/T_1^{(i)})}$, $r_{p}^{(i)} = \exp{(-t_c/T_\phi^{(i)})}$, $T_\phi^{(i)} = T_1^{(i)}T_2^{(i)}/(2T_1^{(i)}-T_2^{(i)})$. $T_1^{(i)}$ and $T_2^{(i)}$ represent the relaxation time and dephasing time of the i-th qubit. The circuit time $t_c$ can be calculated by traversing the DAG of a circuit in the topological order as described in Alg. \ref{alg:circuit_time}.
We integrate $EPST^*$ in our noise-aware initial mapping algorithm in Alg. \ref{alg:im}.

\textbf{Routing}. To be compatible with initial mapping, the routing method should also take noise into account. We use the routing method introduced in \cite{niu2023enabling}. This method improves SABRE routing \cite{li2019tackling} in the following aspects: (1) adding BRIDGE gates as an alternative to SWAP gates, (2) considering the noise of two-qubit gates in the distance matrix, and (3) noticing the impact of inserted SWAP gates and BRIDGE gates themselves.

\subsection{Complexity Analysis}
Given the number of jobs $\nprogram$, the number of gates $g$, the number of physical qubits $\nqpu$, the number of logical qubits $\nqubit$ ($\nqubit<\nqpu$), the number of starting points $m$ ($m<\nqpu$), and the number of repeats $r$, we can calculate the time complexity of our method. The complexity of qubit partitioning and routing is $O\big(m\nqubit^2+\nqpu\log (\nqpu)+g\big)$ and $O(g\nqpu^{2.5})$, respectively \cite{niu2023enabling}.

\textbf{Queue rearranging.} Calculating the priority score takes $O(\nprogram)$ time. The main overhead of queue rearranging lies in sorting the queue, which takes $O\big(\nprogram\log (\nprogram)\big)$ time. Hence, the complexity of queue rearranging is $O\big(\nprogram\log (\nprogram)\big)$.

\textbf{Initial mapping.} The random permutation step takes $O(\nqubit)$ time. In each loop, the routing method takes $O\big(m\nqubit^2+\nqpu\log (\nqpu)+g\big)$, and calculation of $EPST^*$ takes $O(g+\nqubit)$. Hence, the complexity of each loop is $O(rg+r\nqubit+rg\nqpu^{2.5})$. The total complexity of initial mapping can be truncated to $O(rg\nqpu^{2.5}+r\nqubit)$.

Since every job should undergo qubit partitioning, initial mapping, and routing, the total complexity is $O(\nprogram rg\nqpu^{2.5}+\nprogram r\nqubit+\nprogram m\nqubit^2)$. Therefore, the overall time complexity is $O\big(\nprogram\log (\nprogram)+\nprogram rg\nqpu^{2.5}+\nprogram r\nqubit+\nprogram m\nqubit^2\big)$. In normal circumstances, the number of repeats $r$ is a small constant and we have $g>m$, so the complexity can be reduced to $O\big(\nprogram\log (\nprogram)+\nprogram g\nqpu^{2.5}\big)$. The routing overhead $O(\nprogram g\nqpu^{2.5})$ is the dominant part, which is unavoidable. Queue rearranging only incurs trivial overhead compared with routing.
\begin{algorithm}[t]
\small
\SetAlgoLined
 
\KwIn{Partition $P$, Routing Method $Routing()$, Repeat Time $R$, Circuit $\mathcal{C}$, Coupling Graph $\mathcal{G}$, Noise Calibration Data $\mathcal{N}$}
\KwOut{Initial Mapping $Best\_initial\_mapping$}
 $Initial\_mapping = Random\_Permutation(P)$\;
 $Best\_score = 0$\;
 $Best\_initial\_mapping = Initial\_mapping$\;
  \For{$i=1$ to R}{
    $\_, Final\_mapping = Routing(\mathcal{C}, \mathcal{G}, Initial\_mapping, \mathcal{N})$\;
    $\_, Initial\_mapping = Routing(\mathcal{C}, \mathcal{G}, Final\_mapping, \mathcal{N})$\;
    $Routed\_circuit, \_ = Routing(\mathcal{C}, \mathcal{G}, Initial\_mapping, \mathcal{N})$\;
    $Score = EPST^*(Routed\_circuit, Initial\_mapping, \mathcal{N})$\;
    \If{$Score > Best\_score$}{
    $Best\_score = Score$\;
    $Best\_initial\_mapping = Initial\_mapping$}
  }

\caption{Noise-aware Initial Mapping}
\label{alg:im}
\end{algorithm}

\section{Experiments}
\subsection{Protocols}

\textbf{Dataset.} We construct our dataset from RevLib\footnote{RevLib can be accessed at \url{https://www.revlib.org/}. It contains quantum circuits realizing specific gates like a Toffoli gate, arithmetic functions like a 1-bit adder, etc.} \cite{soeken2012revkit}, a benchmark of reversible and quantum circuits, which is widely used in related works \cite{li2019tackling,liu2021qucloud,niu2023enabling}. Circuits with extremely large width or depth are unsuitable for the noise model and real quantum hardware, so we filter the data. First, we choose circuits with width no more than 16 to be suitable for the 16-qubit noise model. Second, we translate circuits to fit the basis gate set of the noise model and Xiaohong. Third, we choose translated circuits with depth smaller than 100 to guarantee relatively high fidelity. Finally, the number of candidate circuits on the noise model and Xiaohong is 77 and 20, respectively. The statistics on our dataset is listed in Tbl. \ref{tab:dataset}.

\begin{table}[ht]
\vspace{-5pt}
\caption{Statistics on our dataset.}
\vspace{-10pt}
\resizebox{\linewidth}{!}{
\begin{tabular}{c|cl|cl|ll|ll}
\hline
\multirow{2}{*}{Environment} & \multicolumn{2}{c|}{$N_{2q}$}       & \multicolumn{2}{c|}{\#Gates (g)}        & \multicolumn{2}{c|}{Depth} & \multicolumn{2}{c}{Width} \\
                             & \multicolumn{1}{l}{range} & avg  & \multicolumn{1}{l}{range} & avg  & range           & avg      & range           & avg     \\ \hline
Noise Model                  & {[}5, 156{]}              & 26.2 & {[}7, 391{]}              & 70.8 & {[}5, 99{]}     & 44.1     & {[}3, 16{]}     & 6.0     \\
Xiaohong                     & {[}5, 22{]}               & 13.6 & {[}19, 111{]}               & 66.6 & {[}15, 68{]}     & 44.1     & {[}3, 16{]}     & 6.7     \\ \hline
\end{tabular}\label{tab:dataset}
}
\vspace{-8pt}
\end{table}

Then, we sample from candidate circuits to construct our dataset. We focus on the congestion scene, where there are some initial jobs and much more jobs to be submitted. 
Due to the limitation of quantum resources and time, the number of initial jobs is 44 on average, and the number of new submitted jobs is 400. 
The submission time $t_i$ of initial jobs is set as 0. For new submitted jobs, $t_{i+1} \in \{t_i, t_i + 1\}$. According to our observation, at peak periods on IBM quantum cloud, there are approximately two jobs submitted per second on average. Hence, the ratio is in line with the congestion in reality. 
For the noise model, the number of shots $s_i$ in each job ($c_i,s_i,t_i$) is set as a random integer from 1K to 20K. For Xiaohong (QuantumCTek), we modify the range of $s_i$ as $\left[500, 10K\right]$ to reduce running overhead. The length of dataset is 10.
\begin{table*}[t!]
\centering
\caption{Performance comparison between different methods (with best in bold and second best underlined.}
\label{tab:all}
\centering
\resizebox{\linewidth}{!}{
\begin{tabular}{l|c|c|c|ccccc|c|c|c}
\hline
\multirow{2}{*}{Environment} &\multirow{2}{*}{Method}  &\multirow{2}{*}{QPU Time[s]$\downarrow$} &\multirow{2}{*}{$\Delta$ QPU Time(\%)$\downarrow$}& \multicolumn{5}{c|}{TAT[s]}            & \multirow{2}{*}{RT[s]$\downarrow$} & \multirow{2}{*}{TRF$\uparrow$}  & \multirow{2}{*}{PST[\%]$\uparrow$} \\ 
   &     &              &      & max$\downarrow$  & avg$\downarrow$  &$\Delta$ avg(\%)$\downarrow$ & std$\downarrow$  &$\Delta$ std(\%)$\downarrow$ &       &      &     \\ \hline
\multirow{5}{*}{\begin{tabular}[c]{@{}c@{}}Noise Model\\ (Guadalupe)\end{tabular}} 
& FIFO       & 925.94                  & 0                 & 5155          & 2591               & 0                & 1484            & 0                & 197    & 1         & \textbf{72.88}    \\
& FIFO-p     & \underline{502.90}                  & \underline{-45.69}            & \underline{2323}           & \underline{1070}               & \underline{-58.71}          & \underline{610}             & \underline{-58.88}           & 208            & 2.43      & 70.06    \\
& NAQJS$^\dagger$       & \textbf{443.76}         & \textbf{-50.07}   & \textbf{2280}  & \textbf{802}       & \textbf{-69.03} & \textbf{585}    & \textbf{-60.61}  & 202           & 2.41      & 69.65    \\
& QuMC       & 734.68                  & -20.66            & 3819          & 2021               & -21.99           & 1046             & -29.52           & 296           & 1.41      & \underline{72.87}    \\
& QuCloud    & 601.62                  & -35.03            & 2897          & 1335               & -48.47           & 785             & -47.11           & 733            & 1.89      & 71.10  \\ \hline
\multirow{5}{*}{\begin{tabular}[c]{@{}c@{}}Noise Model\\ (Chain)\end{tabular}} 
& FIFO      & 925.94            & 0                 & 5155          & 2591          & 0                 & 1484          & 0                & -          & 1               & 69.55    \\
& FIFO-p    & \underline{529.32}            & \underline{-42.83}           & \textbf{2273}   & \underline{1073}          & \underline{-58.57}            & \textbf{632} & \textbf{-57.39}    & -          & 2.30            & 67.68    \\
& NAQJS$^\dagger$      & \textbf{516.40}   & \textbf{-44.23}   & \underline{2413}          & \textbf{876}  & \textbf{-66.21}   & \underline{663}           & \underline{-55.35}            & - & 2.23            & 67.98    \\
& QuMC      & 683.17            & -26.22           & 3405          & 1611          & -37.84            & 936           & -36.90           & -         & 1.60            & \underline{69.66}    \\
& QuCloud   & 721.09            & -22.12           & 3688          & 1873          & -27.70            & 1052           & -29.14           & -          & 1.46            & \textbf{69.90}  \\ \hline
\multirow{6}{*}{Xiaohong} & FIFO & 468.31       & 0                 & 4733          & 2377              & 0       & 1361 & 0       & 110    & 1            & \textbf{45.86}   \\
& FIFO-p    & \textbf{95.89}       & \textbf{-79.52}            & \textbf{455} & \underline{217}  & \underline{-90.88}  & \textbf{118}  & \textbf{-91.33}  & 112    & 8.27 & 32.31   \\
& NAQJS$^\dagger$    & \underline{96.75}    & \underline{-79.34}            & \underline{688} & \textbf{156}  & \textbf{-93.45}  & \underline{176}  & \underline{-87.09}  & 133    & 6.62 & 35.70   \\
& NAQJS$^\dagger$ ($\eta=2/7$)    & 216.21    & -53.83            & 1826 & 594  & -74.99  & 484  & -64.41  & 125    & 2.53 & 42.48   \\
& QuMC    & 270.78      & -42.18 & 2177 & 917  & -61.43  & 627  & -53.90  & 515    & 2.31 & 43.33   \\
& QuCloud    & 372.49       & -20.46            & 3637 & 2025 & -14.82  & 1191 & -12.45  & 1526   & 1.34 & \underline{43.97}   \\ \hline
\end{tabular}
}
    \begin{scriptsize}
    \begin{tablenotes}
\item 
\textbf{$\dagger$:} our proposed method.
\textbf{$\Delta$ QPU Time(\%):} percentage difference to the QPU time of FIFO. 
\textbf{TAT[s]:} turnaround time in seconds. 
\textbf{RT[s]:} runtime of scheduling algorithms in seconds. 
\textbf{$\Delta$ avg(\%):} percentage difference to the average of FIFO.
\textbf{$\Delta$ std(\%):} percentage difference to the standard deviation of FIFO. 
\textbf{TRF:} Trial Reduction Factor \cite{das2019case}.
 \end{tablenotes}
 \end{scriptsize}
\vspace{-10pt}
\end{table*}

\textbf{Baselines.} As there are few existing methods addressing the scheduling problem proposed in this paper, we devise four baselines (i.e. FIFO, FIFO-p, QuMC, QuCloud) to verify the effectiveness of NAQJS. First-In-First-Out (FIFO) denotes the current running mode of quantum computers. Specifically, each submitted quantum job is executed in serial according to their submission time. FIFO-p represents that all quantum jobs are executed in parallel according to their submission time. In each execution round, quantum jobs will be allocated on the quantum processor in chronological order until the next job cannot be accommodated. The partition and mapping methods of FIFO and FIFO-p are the same as NAQJS. Additionally, we adapt QuMC \cite{niu2023enabling} and QuCloud \cite{liu2021qucloud} to fit QJSP by merging their queuing method into our framework. For baselines, we do not explicitly restrain the maximum usage of physical qubits.

\textbf{Parameter Setting and Experiment Environment}. Experiments are performed on the IBM Guadalupe noise model and its chain version (discussed in Sec. \ref{sec:topo}) as simulation, and the physical Xiaohong (QuantumCTek) quantum processor. For Xiaohong, the average relaxation time $\overline{T}_1$ and $\overline{T}_2$ are 27.35 $\mu s$ and 20 $\mu s$. The average reliability of one-qubit gate $\overline{r}_{1q}$, two-qubit gate $\overline{r}_{2q}$, and measurement $\overline{r}_{ro}$ are 99.85\%, 97.07\%, and 93.97\%.

For noise models, we set $\alpha=6$, $\beta=4.5$, $\gamma=1$, $\eta=5/6$, $\Delta t=360$. For Xiaohong, we set $\alpha=6$, $\beta=3$, $\gamma=1$, $\eta=5/6$, $\Delta t=360$. According to expert knowledge and our practical tests on quantum processors, we set the time cost of every shot as 200$\mu s$, and extra time between execution as 10$s$.

\subsection{Results on the Noise Model Guadalupe}

As shown in Tbl. \ref{tab:all}, our method NAQJS achieves the shortest average turnaround time (TAT) across all methods. It reduces TAT of FIFO by nearly 70 percent, which will significantly cut down the waiting time for users to obtain their results. Also, the standard deviation of TAT of NAQJS is the smallest, having a reduction of 60.61\% over FIFO. Small standard deviation means TAT of different users will not differ too much, which showcases the fairness of NAQJS. The maximum TAT of NAQJS ranks second (only a bit longer than FIFO-p), indicating that no job will wait too long to be executed, further strengthening the fairness of our method. Besides, our QPU time is the shortest among all the five methods, which attains 50.07\% reduction over FIFO. The PST reduction of NAQJS is only 3.23\%. In other words, NAQJS can achieve significant improvements in QPU time and TAT at a trivial cost of fidelity. 

Compared with other methods, NAQJS is superior in QPU time and TAT. Though QuMC can ensure high PST, the QPU time and TAT are about twice longer than NAQJS. The improvements of QuMC in time metrics over FIFO is rather limited. 

\subsection{Results on Xiaohong Quantum Processor}
As shown in Tbl. \ref{tab:all}, we still achieve the shortest average TAT, significantly decreasing TAT of FIFO by 93.45\%. Among all methods, NAQJS has the second lowest QPU time (79.34\% reduction over FIFO) and standard deviation of TAT (87.09\% reduction over FIFO). Hence, NAQJS can significantly reduce time overhead for both users and suppliers, and meanwhile ensure enough fairness. 

The superiority of NAQJS on Xiaohong owes to the large TRF (6.62) \cite{das2019case} and our queue rearranging. TRF is the ratio of the number of trials when circuits run in series to that when circuits run in parallel. With 66 physical qubits, Xiaohong can accommodate more jobs in each execution. Hence, the execution times are largely diminished compared to FIFO, resulting in shorter QPU time and TAT. In addition, our priority score can perceive the potential influence of each job on time metrics, arranging those highly influential jobs to the head of the queue. Hence, QPU time and TAT are further reduced. Though FIFO-p slightly outperforms NAQJS in QPU time by 0.18 percent due to larger TRF (8.27), its average TAT is 2.57 percent larger than ours, and its PST is 3.39 percent lower than us. Considering all these metrics, NAQJS performs the best in general.

PST on Xiaohong is much lower, because noise on Xiaohong is more severe than on the noise model. NAQJS has 10.16\% reduction in PST over FIFO while QuMC and QuCloud keep relative high PST (43.33\% and 43.97\%). However, their TRF is only 2.31 and 1.34. Hence, their QPU Time and average TAT are more than twice longer than ours. We further validate that NAQJS surpasses QuMC and QuCloud in time metrics by a large gap even when PST is close. When we set the maximum usage $\eta$ as 2/7 in Tbl. \ref{tab:all}, the PST (42.48\%) is almost the same as QuMC and QuCloud, but the QPU time and average TAT are still much lower than theirs (over 10\%). 

\subsection{Runtime Analysis}
As shown in Tbl. \ref{tab:all}, the runtime (RT) of NAQJS is close to FIFO. As routing occupies most of the runtime, NAQJS will not introduce much additional overhead. QuMC and QuCloud cost much more time than us, especially on Xiaohong. QuMC and QuCloud will repeat routing if a job is unsuitable for their strategy. The runtime of them increases dramatically with the growth in the number of physcal qubits, while NAQJS avoids this issue, indicating the scalability. Note that NAQJS can run in a pipelining manner with the circuit execution, and the QPU time plus the total extra time (2279.76 s for Guadalupe and 780.09 s for Xiaohong) is far longer than our runtime. Hence, it will not affect the running of QPU. 

\subsection{Impact of QPU Topology}\label{sec:topo}
Quantum processors may have different topology, leading to distinct connectivity of qubits. The topology of the noise model Guad\-alupe is ring-type shown in Fig. \ref{fig:qpu}b. To explore the impact of QPU Topology on QJSP metrics, we disconnect $Q_1$ from $Q_4$ in Guadalupe and obtain chain-type topology depicted in Fig. \ref{fig:chain}. The noise information of every qubit remains unchanged. The results are listed in Tbl. \ref{tab:all}. The connectivity of chain-type topology (Chain) is worse than ring-type topology (Guadalupe), so the performance of all the methods deteriorates. NAQJS still outperforms others in both QPU time and average TAT. Besides, the PST gap between our method and FIFO gets smaller.
\begin{figure}[htbp]
    \centering
    \includegraphics[width=\linewidth]{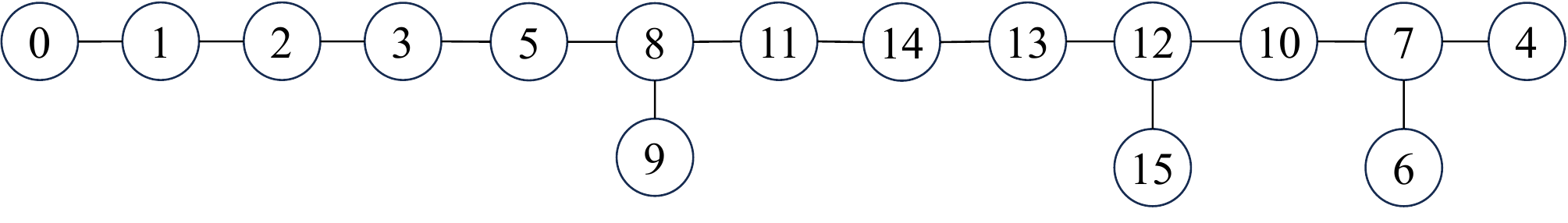}
    \caption{Chain-type topology.}
    \label{fig:chain}
\vspace{-12pt}
\end{figure}
\begin{figure*}[t]
\centering
	\begin{subfigure}[b]{0.33\linewidth}
	\includegraphics[width=\linewidth]{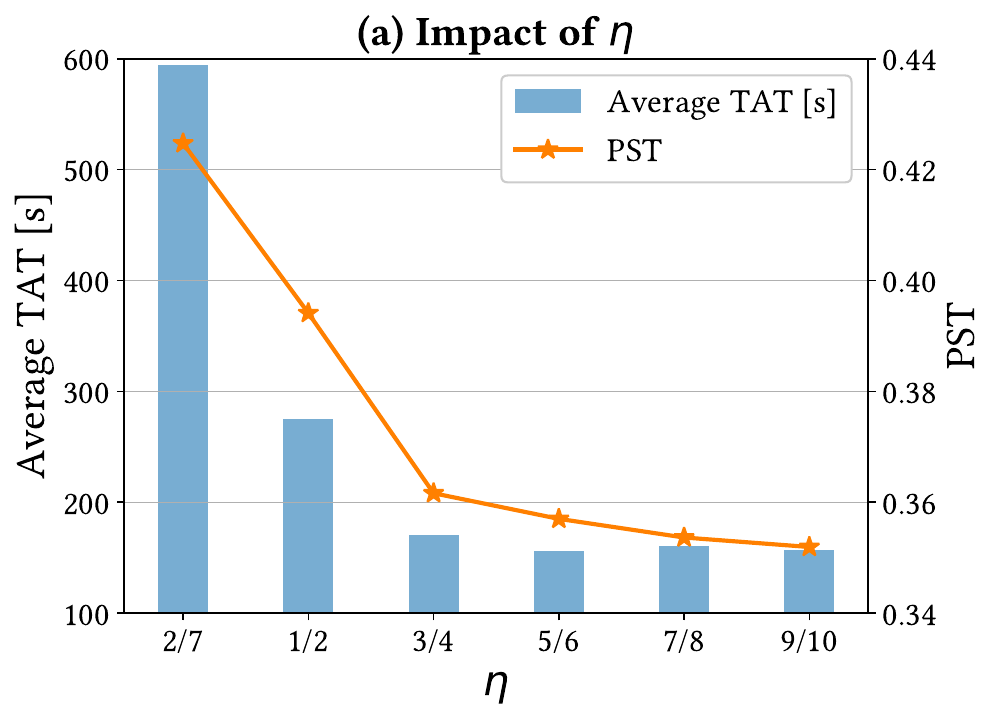}
\label{fig:eta}
	\end{subfigure}
	\begin{subfigure}[b]{0.33\linewidth}
    \includegraphics[width=\linewidth]{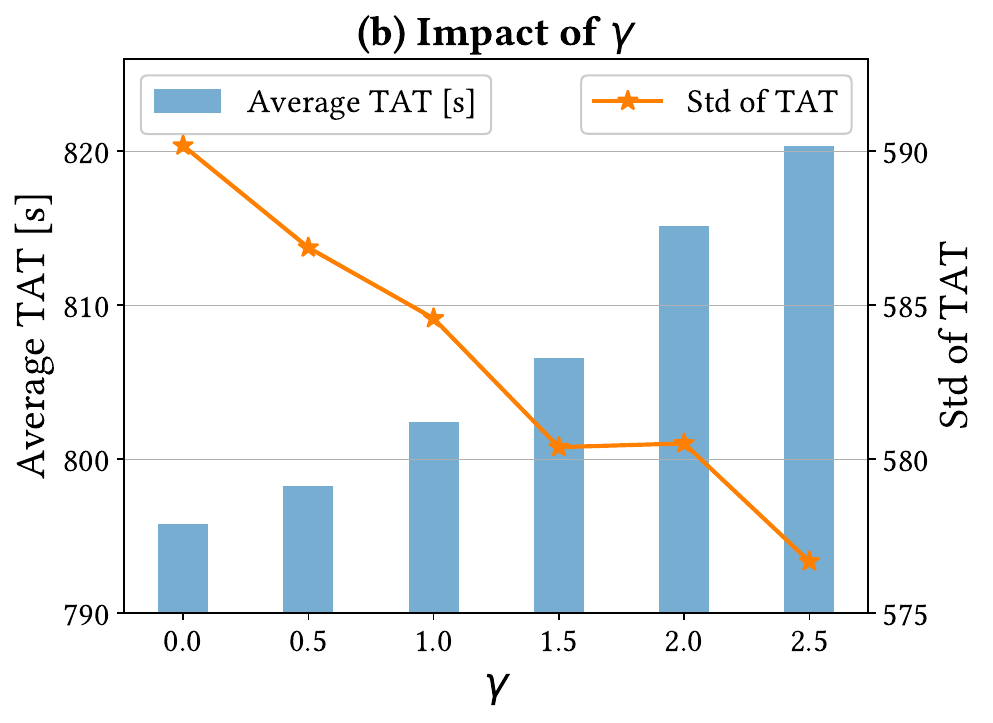}
\label{fig:gamma}
\end{subfigure}  
	\begin{subfigure}[b]{0.33\linewidth}
	\includegraphics[width=\linewidth]{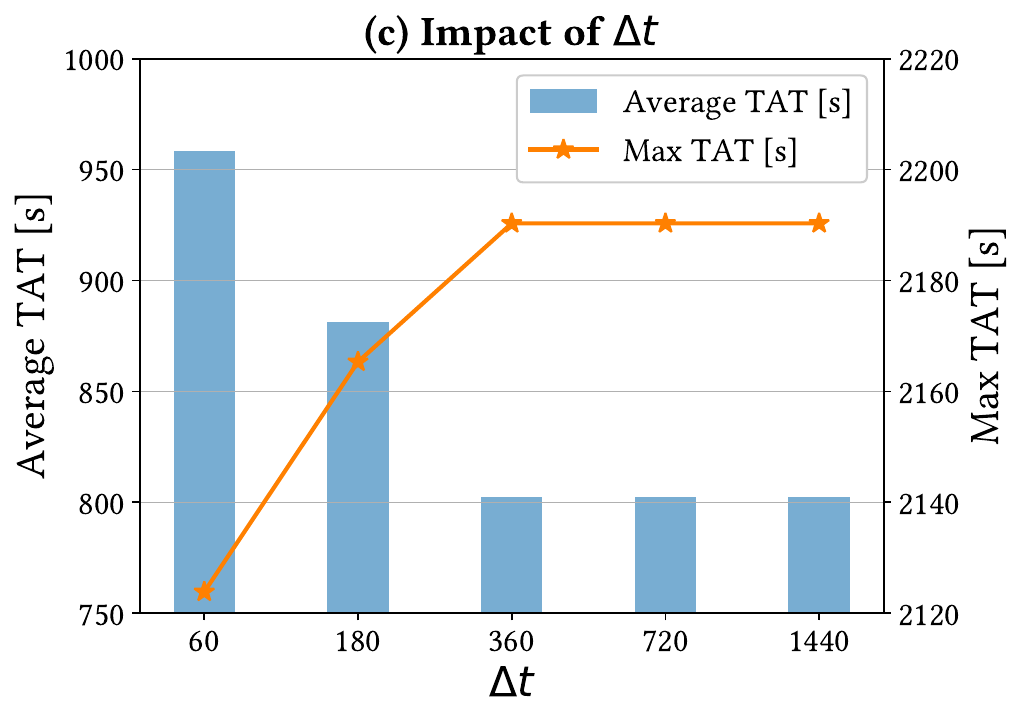}
\label{fig:deltat}
	\end{subfigure}
\vspace{-30pt}
\caption{(a) Impact of $\eta$ on average TAT and PST on Xiaohong. (b) Impact of $\gamma$ on average TAT and standard deviation of TAT on the noise model. (c) Impact of $\Delta t$ on average TAT and maximum TAT on the noise model.}\label{fig:sensitivity}
\vspace{-8pt}
\end{figure*}
\begin{figure*}[htbp]
\centering
	\begin{subfigure}[b]{0.31\linewidth}
	\includegraphics[width=\linewidth]{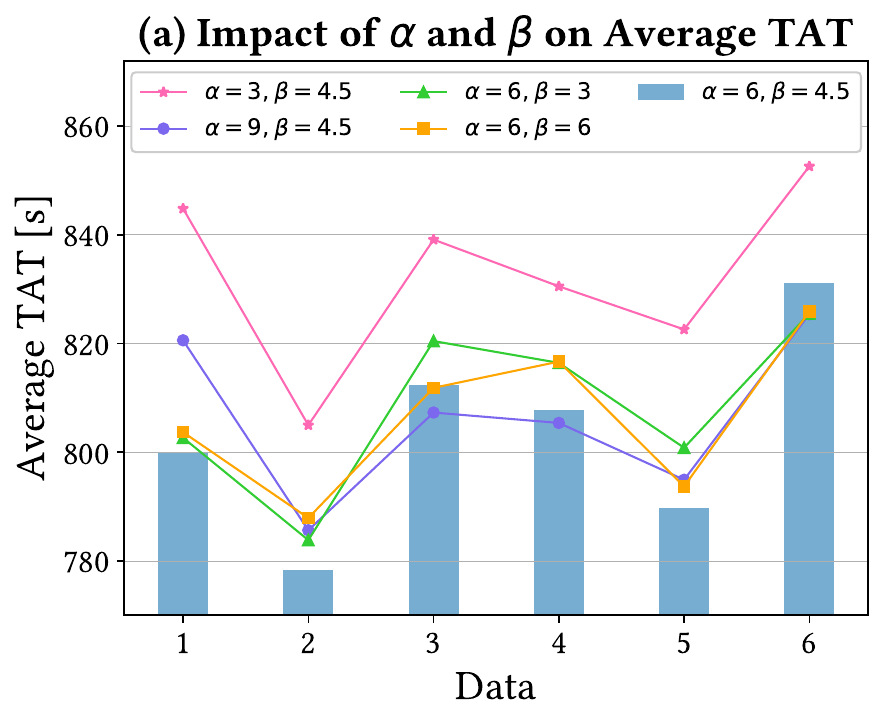}
	\end{subfigure}
	\begin{subfigure}[b]{0.31\linewidth}
    \includegraphics[width=\linewidth]{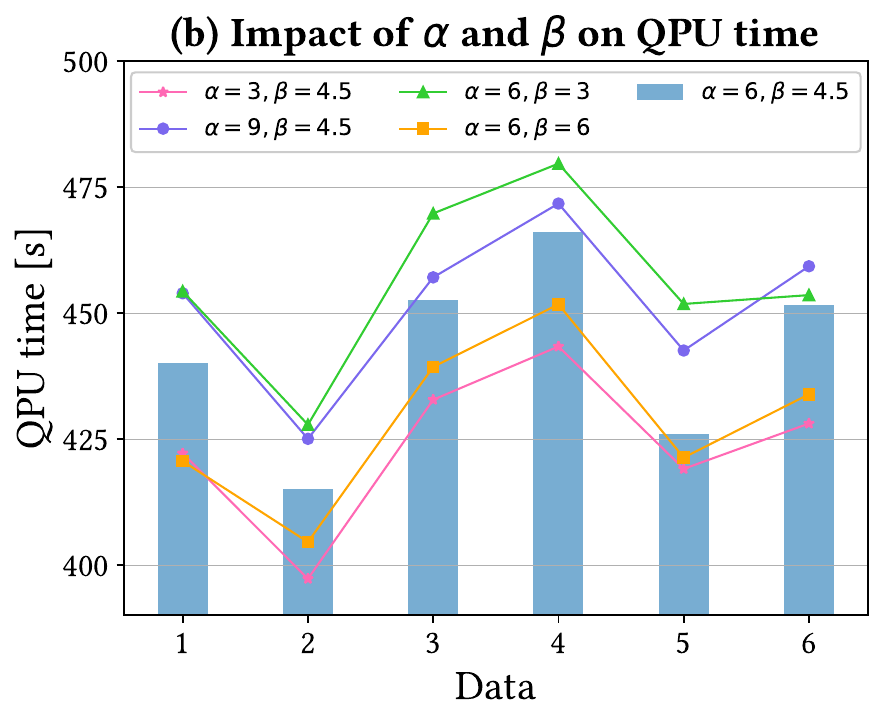}
\end{subfigure} 
	\begin{subfigure}[b]{0.35\linewidth}
    \includegraphics[width=\linewidth]{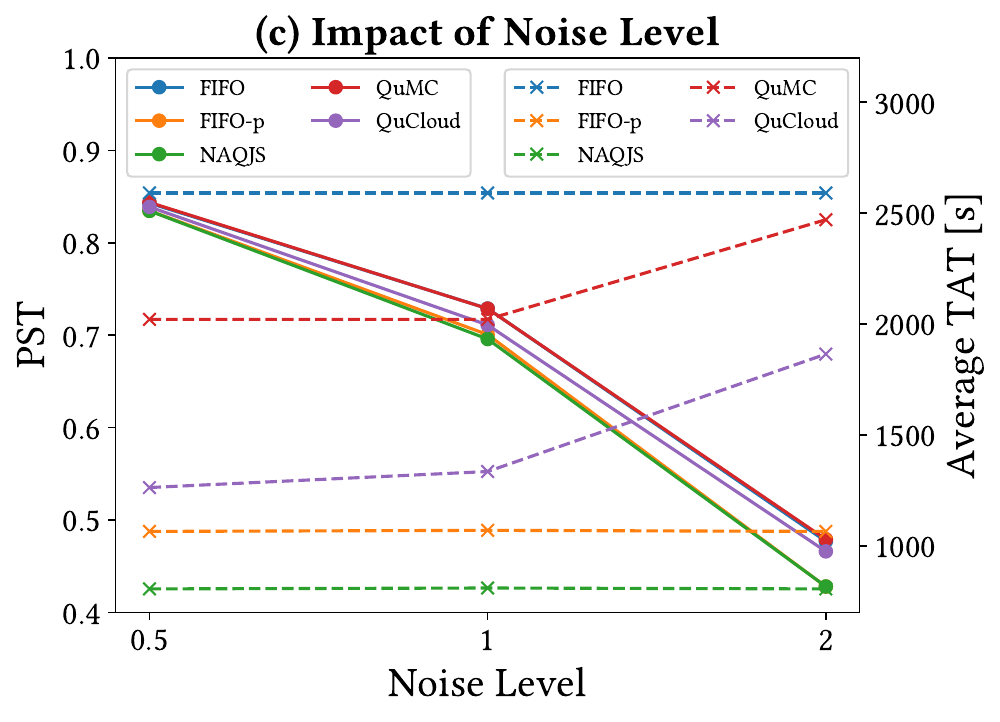}
\end{subfigure}  
\vspace{-12pt}
\caption{(a) Impact of $\alpha$ and $\beta$ on average TAT on the noise model. (b) Impact of $\alpha$ and $\beta$ on QPU time on the noise model. (c) Impact of noise level on PST (solid lines) and average TAT (dashed lines) on the noise model.}\label{fig:alpha_beta}
\vspace{-8pt}
\end{figure*}

\subsection{Impact of Noise Level}
The noise level can affect the metrics in QJSP, especially PST. To investigate this, we multiply the noise of one-qubit gates, two-qubit gates and measurements by the noise level (0.5, 1, or 2). As shown in Fig. \ref{fig:alpha_beta}c, PST of all the methods drops drastically with the increase of the noise level. QuMC and QuCloud sacrifice average TAT for PST when noise level is large. QuMC almost degrades to the serial running mode (one job per execution) when noise level reaches 2. By contrast, average TAT of NAQJS is insensitive to the noise level. Therefore, NAQJS can keep low average TAT and adequately high PST even when the noise condition is poor.

More importantly, the PST gap between NAQJS and other methods narrows when noise level decreases. Researchers are currently devoted to manufacturing larger-scale quantum processors and fabricating noiseless logical qubits. As a result, users will pay more and more attention to time metrics in the future. The advantage in time metrics of NAQJS will be further amplified with the increase in qubit number and decrease in noise.

\subsection{Sensitivity Analysis of Hyperparameters}\label{sec:sensitivity}
We study the impact of five hyperparameters: width weight $\alpha$, shot weight $\beta$, time weight $\gamma$, maximum usage $\eta$, and time interval $\Delta t$. The sensitivity tests for $\alpha$, $\beta$, $\gamma$, and $\Delta t$ are implemented on the noise model; the test for $\eta$ is done on Xiaohong, as $\eta$ is more related with PST, and PST on Xiaohong reflects the realistic effect on QPUs.

\textbf{Maximum usage $\eta$}. $\eta$ directly influences TRF and the qubit utilization rate. With its increase, both average TAT and PST decline as shown in Fig. \ref{fig:sensitivity}a, because more jobs are executed in parallel. Hence, we can balance the fidelity and execution latency by tuning $\eta$ and we set $\eta=5/6$. 

\textbf{Time weight $\gamma$.} As shown in Fig. \ref{fig:sensitivity}b, with the increase of $\gamma$, the standard deviation of TAT is getting smaller while average TAT is getting longer. Hence, there is a trade-off between fairness and time metrics. Since time metrics are more important for both users and suppliers, we pick $\gamma=1.0$ in our experiments.

\textbf{Width weight $\alpha$ and shot weight $\beta$.} $\alpha$ and $\beta$ are more concerned with TAT and QPU time. As shown in Fig. \ref{fig:alpha_beta}a and Fig. \ref{fig:alpha_beta}b, the increase in $\alpha$ leads to decline in average TAT and rise in QPU time. By contrast, larger $\beta$ reduces QPU time, and average TAT reaches its minimum when $\beta=4.5$. To strike a balance between average TAT and QPU time, we pick $\alpha=6, \beta=4.5$ for the noise model.

\textbf{Time interval $\Delta t$.} We introduce $\Delta t$ to avoid the starvation problem. $\Delta t$ should be big enough, or it will result in FIFO-p manner execution. As shown in Fig. \ref{fig:sensitivity}c, with the increase of $\Delta t$, average TAT decreases while maximum TAT rises, indicating that quantum jobs with large width or shot numbers waits longer to be executed. Nevertheless, the average TAT and maximum TAT saturates at $\Delta t=360$. Therefore, we set $\Delta t=360$ to avoid the starvation problem and guarantee low average TAT.

\subsection{Ablation Study of EPST*}
To evaluate our $EPST^*$ score, we run NAQJS with and without $EPST^*$. For the noise model, $EPST^*$ raises PST by 0.70\% (from 68.95\% to 69.65\%). For Xiaohong, it raises PST by 1.77\% (from 33.93\% to 35.70\%). Hence, $EPST^*$ can improve the fidelity of quantum circuits.

\section{Conclusion and Future Work}
We have formulated the Quantum Job Scheduling Problem (QJSP) and proposed a noise-aware quantum job scheduler (NAQJS) to boost the execution efficiency of (superconducting) quantum processors. Our scheduling method perceives the impact of different jobs on time metrics through our priority score, and the noise-aware initial mapping improves the fidelity. Results show that NAQJS outperforms all the baselines in both QPU time and TAT. Besides, the fidelity and fairness are also guaranteed. The small runtime overhead shows its scalability on large QPUs. QJSP may be more important when noise-free quantum processors emerges, because it will directly influence the efficiency of them.

\textbf{Future Work.} In this paper, we conduct experiments on a 66-qubit quantum processor, which is the largest in scale compared with related works. Since larger-scale QPUs are not yet open to the public or too expensive to use, we leave larger-scale experiments for our future work. Besides, we envision that NAQJS may be further adapted to non-superconducting quantum cloud.


\bibliographystyle{ACM-Reference-Format}
\bibliography{main}

\end{document}